\documentclass[aps,twocolumn,showpacs,amsmath,amssymb,floatfix]{revtex4}
\usepackage{amsthm}
\usepackage{latexsym}
\usepackage{amsfonts}
\usepackage{bbm,dsfont}
\usepackage{graphicx}

\newtheorem{proposition}{Proposition}

\newcommand{\h}[1]{\mathcal{#1}}
\newcommand{\R}{\mathbb{R}}

\newcommand{\N}{\mathbb{N}}
\newcommand{\hil}{\mathcal{H}}

\newcommand{\lh}{\mathcal{L(H)}}

\newcommand{\E}{\mathsf{E}}

\newcommand{\G}{\mathsf{G}}

\begin{document}
\title{Balancing efficiencies by squeezing in realistic eight-port 
homodyne detection}
\author{Juha-Pekka Pellonp\"a\"a}
\email{juha-pekka.pellonpaa@utu.fi}
\affiliation{Turku Centre for Quantum Physics, 
Department of Physics and Astronomy, University of Turku, FI-20014 Turku, Finland}
\author{Jussi Schultz}
\email{jussi.schultz@utu.fi}
\affiliation{Turku Centre for Quantum Physics, Department of 
Physics and Astronomy, University of Turku, FI-20014 Turku, Finland}
\author{Matteo G. A. Paris}
\email{matteo.paris@fisica.unimi.it}
\affiliation{Dipartimento di Fisica, Universit\`a degli Studi di Milano, 
I-20133 Milano, Italy}
\affiliation{CNISM, UdR Milano, I-20133 Milano, Italy}
\date{\today}
\pacs{}
\begin{abstract}
We address measurements of covariant phase observables (CPOs) by means of
realistic eight-port homodyne detectors. We do not assume equal quantum
efficiencies for the four photodetectors and investigate the conditions
under which the measurement of a CPO may be achieved.  We show that
balancing the efficiencies using an additional beam splitter allows us
to achieve a CPO at the price of reducing the overall effective
efficiency, and prove that it is never a smearing of the ideal CPO
achievable with unit quantum efficiency.  An alternative strategy based
on employing a squeezed vacuum as a parameter field is also suggested, which
allows one to increase the overall efficiency in comparison to the
passive case using only a moderate amount of squeezing. Both methods are
suitable for implementantion with current technology.
\end{abstract}
\pacs{42.50.Ar, 42.50.Ct, 03.65.-w}
\maketitle
\section{Introduction}
In quantum mechanics, the concept of phase for a radiation mode 
has always remained a somewhat controversial topic, with both
fundamental and technological implications, see \cite{r0,r1,r2,r3,r4,r5,Pel02} 
and references therein for a review. A major reason for this is that 
in trying to define the phase of a quantum oscillator one can clearly 
see the restrictions of the conventional
approach which identifies observables as self-adjoint operators, or
equivalently, their spectral measures. Indeed, it can be shown that no
spectral measure satisfies  the physically relevant conditions posed on
phase observables \cite{Hol83,Lahti,boh95,can96}. However, this problem has been
overcome with the introduction of the more general concept of
observables as positive operator measures. In this approach the
concept of a covariant phase observable (CPO) naturally emerges and 
these observables have been completely charaterized \cite{Hol83,Lahti}.  An important class
of CPOs arise as the angle margins of certain covariant phase space
observables, the most familiar example being the $Q$-function of the
field. Their physical significance is further emphasized by the fact
that any phase space observable can in principle be measured via
eight-port homodyne detection, a method which was introduced in the
microwave domain \cite{wal86} and then extensively analyzed in the
optical domain
\cite{wal87,Lay89,Man9X,Fre93,Pau93,Dar94,Lui96,Wod97,Hra00}. 
Other multiport homodyne \cite{Ray93,Par97} and heterodyne detection 
\cite{Yue78,Sha84,Sha85} may be employed as well, the latter also in 
the presence of frequency mismatch \cite{ZGM}.
\par
Any realistic measurement is subject to noise due to imperfections in
the measuring apparatus.  In the case of eight-port homodyne detection,
one of the relevant sources of noise is the presence of detector
inefficiencies. Indeed, as reported in \cite{Had09}, the quantum
efficiencies of commercially available detectors range from very high to
as low as a few percents and their effect is far from being 
negligible. In eight-port homodyning the presence of detector 
inefficiencies causes a Gaussian smearing on the measured observable 
\cite{Leo93,Opt95}. This appears as a convolution
structure which causes the actually measured distributions to be
smoothed versions of the ideal ones. As a matter of fact, 
quantum efficiencies of the photodetectors are traditionally assumed 
to be equal which results in a rotation invariant convolving measure. In other
words, the smoothing effect is the same in any direction in the phase
space. A detailed analysis shows that this symmetry is lost if we drop
the assumption  of equal efficiencies \cite{Lahti2}. This loss of
symmetry is crucial when the measurement is intended to gain information
about the phase properties of the field, and it is the purpose of this
paper to address this problem in detail.  
\par
We consider two methods for regaining this lost symmetry. At first, we show 
that the efficiencies can be balanced by inserting an additional beam 
splitter in front of one of the photodetectors. This results in a 
decreased overall efficiency for the measurement scheme. We also show 
that the angle margin of the measured phase observable is never a smearing 
of the ideal one. We then consider the effect of squeezing the parameter field 
while keeping the efficiencies fixed. As it turns out, this also compensates 
the efficiencies mismatch, thus retrieving the lost symmetry. We 
also compare the two methods and show that the overall efficiency is always 
greater for the squeezing strategy.
\par
The paper is organized as follows. In Section \ref{observables} we lay out 
the general framework and give the necessary definitions. Section \ref{detector} 
is devoted to the mathematical description of eight-port homodyne detection involving 
non-ideal photodetectors. In Sections \ref{beamsplitter} and \ref{squeezing} 
we describe in some details the aforementioned methods of overcoming the problems 
arising from different quantum efficiencies. The conclusions and future outlooks 
are presented in Section \ref{conclusions}.
\section{Covariant phase observables and phase space observables}\label{observables}
Let $\hil$ be the infinite dimensional separable Hilbert space
associated with a single mode electromagnetic field, and let $\lh$
denote the set of bounded  operators acting on $\hil$. We fix the photon
number basis $\{ \vert n\rangle \vert n=0,1,2,\ldots \}$ and denote by
$N$ the number operator associated with this basis. By diagonality of a
bounded operator we always mean diagonality with respect to the number
basis. We will use without explicit indication the coordinate
representation in which case $\hil$ is identified with $L^2 (\R)$ and
the basis vectors with Hermite functions. The states of the field are
represented by positive operators with unit trace, and the observables
are represented by normalized positive operator measures $\E :\h
B(\Omega)\rightarrow \h L(\hil)$, where $\h B(\Omega)$ stands for the
Borel $\sigma$-algebra of subsets of the measurement outcome space
$\Omega$. For a field in a state $\rho$, the measurement outcome
statistics of an observable $\mathsf{E}$ is given by the probability
measure $X\mapsto \textrm{tr} [\rho \mathsf{E} (X)]$. 
\par
An observable $\mathsf{\Phi}:\h B([0,2\pi))\rightarrow \lh$ is a {\em
covariant phase observable} (CPO) if 
\begin{equation*}
e^{i\phi N}\mathsf{\Phi} (X) e^{-i\phi N} = \mathsf{\Phi}(X \dot{+}\, \phi)
\end{equation*}
for all $X\in \h B([0,2\pi))$ and $\phi \in [0,2\pi)$, where $\dot{+} $ 
denotes addition modulo $2\pi$. According to the Phase Theorem 
\cite[Theor. 2.2]{Lahti}, each phase observable is of the form
\begin{equation}\label{phase}
 \mathsf{\Phi}(X) =\sum_{m,n=0}^\infty c_{mn}\, 
 \frac{1}{2\pi} \int_X e^{i(m-n)\alpha} \, d\alpha\, 
 \vert m\rangle\langle n\vert,
\end{equation}
for some unique {\em phase matrix} $(c_{mn})_{m,n=0}^\infty$, that is, 
a positive semidefinite complex matrix 
satisfying $c_{nn}=1$ for all $n\in\N$. The phase observables measured
by eight-port homodyne detection arise as angle margins of certain 
covariant phase-space observables. 
\par
An observable $\G:\h B(\R^2)\rightarrow \lh$ is a {\em covariant 
phase-space observable} if 
\begin{equation*}
W(q,p) \G (Z) W(q,p)^* =\G (Z+ (q,p))
\end{equation*}
for all $Z\in\h B(\R^2)$ and $(q,p)\in\R^2$, where 
$W(q,p)=e^{i\frac{qp}{2}}e^{-iqP}e^{ipQ}$ are the Weyl operators. 
Any covariant phase-space observable is generated by a unique positive 
unit trace operator $\sigma$ so that the observable is of the form 
\cite{Holevo, Werner}
\begin{equation*}
\G^\sigma (Z) =\frac{1}{2\pi}\int_Z W(q,p) \sigma W(q,p)^*\, dqdp
\end{equation*}
Now let us denote by $\mathsf{\Phi}^\sigma:\h B([0,2\pi))\rightarrow\lh$ 
the angle margin of $\G^\sigma$, that is, 
\begin{equation*}
\mathsf{\Phi}^\sigma(X) =\G^\sigma (X\times [0,\infty)),\qquad X\in\h B([0,2\pi)),
\end{equation*}
where the relation between the polar and Cartesian coordinates is given by 
$re^{i\alpha} =\frac{1}{\sqrt{2}} (q+ ip)$. The key result needed in our study 
is \cite[Theor. 4.1]{Lahti} which states that {\em $\mathsf{\Phi}^\sigma$ is 
a phase observable if and only if $\sigma$ is diagonal}. The simplest and from the experimental point of view the most useful example is the case $\sigma =\vert 0\rangle\langle 0\vert$, that is, when the observable is generated by the vacuum state. In this case the phase distribution is just the angle margin of the Husimi $Q$-function of the field. 
\par
It should be stressed that even though CPOs arise naturally as the margins of covariant phase space observables, not all CPOs are obtained in this way. In particular, the canonical phase observable is not the angle margin of any phase space observable \cite{Lahti}. 
In order to go into the analysis of phase observables related to eight-port 
homodyne detection, we need to recall the details of the measurement 
scheme. This is the subject of the next Section.
\section{Eight-port homodyne detector}\label{detector}
The eight-port homodyne detector consists of four input modes, four
balanced $50$:$50$ beam splitters, a phase shifter which provides 
a phase-shift of $\frac{\pi}{2}$ on one of the modes and four 
photodetectors with quantum efficiencies $\epsilon_j$, $j=1,2,3,4$ 
(see Fig. \ref{eightport}) which are not assumed to be equal. 
The measured quantities are the suitably scaled photon number 
differences between modes 1 \& 3, and 2 \& 4, respectively. The
signal field in mode 1 is the field under investigation while the
parameter field in mode 2 determines the measured observable. The input
mode 3 is left empty so it corresponds to a vacuum field and the local
oscillator in mode 4 is in a coherent state $\vert \sqrt{2} z\rangle$. 
The procedure for obtaining the phase distribution with this setup can be 
described as follows. Each
experimental event consists of a simultaneous detection of the
two commuting difference-photocurrents  which trace a pair of
field-quadratures. Each event thus corresponds to a point in the 
complex plane  and the phase value 
inferred from the event is the polar angle of the point itself. 
The experimental histogram of the phase distributions is
obtained upon dividing the plane into \lq\lq infinitesimal\rq\rq\ 
angular bins of equal width, from $0$ to $2\pi$,
then counting the number of points which fall into each 
bin. We shall next go into the mathematical description in more detail.
\begin{figure}[h!]
\includegraphics[width=0.95\columnwidth]{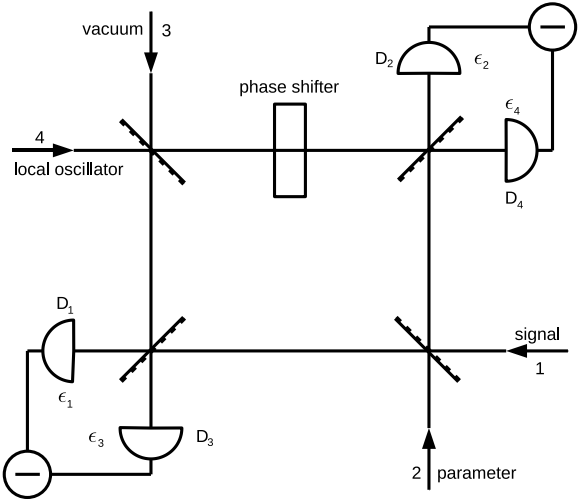}
\caption{Schematic diagram of the eight-port homodyne detection scheme.
The scheme consists of four input modes, four balanced $50$:$50$ beam 
splitters, a phase shifter which provides a phase-shift of $\frac{\pi}{2}$ 
on one of the modes and four photodetectors with quantum efficiencies 
$\epsilon_j$, $j=1,2,3,4$ which are not assumed to be equal. 
The measured quantities are the photon number differences between modes 
1 \& 3, and 2 \& 4, respectively, rescaled by the amplitude 
of the local oscillator, i.e. a strong coherent state $\vert \sqrt{2} z\rangle$
impinged into mode 4. The signal field in mode 1 is the field under 
investigation, while the parameter field in mode 2 determines 
the measured observable. The input mode 3 is left empty so it 
corresponds to a vacuum field.}\label{eightport}
\end{figure}
\par
In order to obtain measurements of covariant phase-space observables, we need to take the high-amplitude limit, that is, assume a very strong local oscillator. Indeed, if
$\sigma'$ is the state of the parameter field and we assume ideal
detectors ($\epsilon_j=1$ for all $j$), the measured observable in the
high-amplitude limit $\vert z\vert \rightarrow \infty$ is
$\G^{\sigma}$, where the generating operator is $\sigma = C\sigma'
C^{-1} $; here $C$ denotes the conjugation map $(C\psi)(x)
=\psi(x)^*\,$ \cite{Kiukas}.
The presence of detector inefficiencies causes a Gaussian smearing 
so that the actually measured observable is given by 
$\mu_{\epsilon_{13},\epsilon_{24}} *\G^\sigma:\h B(\R^2)\rightarrow \lh$ 
defined as 
\begin{equation*}
(\mu_{\epsilon_{13},\epsilon_{24}} *\G^\sigma) (Z) = \int \mu_{\epsilon_{13},
\epsilon_{24}} (Z-(q,p))\, d\G^\sigma (q,p),
\end{equation*}
where $\mu_{\epsilon_{13},\epsilon_{24}}:\h B(\R^2)\rightarrow [0,1]$ is 
a probability measure with the density
\begin{align}
(q,p)\mapsto &\: \frac{1}{2\pi}
\sqrt{\frac{\epsilon_{13}\epsilon_{24}}{(1-\epsilon_{13})(1-\epsilon_{24})}} 
\notag \\ & \times
\exp\left\{-\frac{\epsilon_{13}}{2(1-\epsilon_{13})} q^2-
\frac{\epsilon_{24}}{2(1-\epsilon_{24})} p^2\right\}\notag\,,
\end{align}
where $ \epsilon_{ij} =\tfrac{2\epsilon_i \epsilon_j}{\epsilon_i  
+\epsilon_j}$ \cite{Lahti2}. The quantities $\epsilon_{13}$ and 
$\epsilon_{24}$ may be viewed as overall efficiencies related to 
the two balanced homodyne detectors in the scheme. In particular,
$$\min \{ \epsilon_i, \epsilon_j \} \leq \epsilon_{ij} 
\leq \max \{ \epsilon_i, \epsilon_j \}\,.$$
The smeared phase-space observable is still covariant and thus generated by 
some positive trace one operator. Indeed, we have $$\mu_{\epsilon_{13},
\epsilon_{24}} *\G^\sigma=\G^{\mu_{\epsilon_{13},\epsilon_{24}}
*\sigma}\,,$$ where $\mu_{\epsilon_{13},\epsilon_{24}} *\sigma$ is the 
convoluted state \cite{Werner}
$$\mu_{\epsilon_{13},\epsilon_{24}} *\sigma =\int W(q,p) 
\sigma W(q,p)^*\, d\mu_{\epsilon_{13},\epsilon_{24}}(q,p)\,.$$ 
The angle margin of the measured phase space observable is then 
$\mathsf{\Phi}^{\mu_{\epsilon_{13},\epsilon_{24}} *\sigma}$ and 
the problem is to determine the conditions under which this is a 
CPO. In other words, we need to determine when the 
generating operator is diagonal. At first we give a partial characterization 
in the following Proposition. 
\begin{proposition}\label{diagonal}
If $\sigma$ is diagonal, then $\mu_{\epsilon_{13},\epsilon_{24}} 
*\sigma$ is diagonal if and only if  $\epsilon_{13} =\epsilon_{24}$. 
Conversely, if $\epsilon_{13} =\epsilon_{24}$, then $\mu_{\epsilon_{13},
\epsilon_{24}} *\sigma $ is diagonal if and only if $\sigma $ is diagonal.
\end{proposition}
\begin{proof}
First notice that any two trace class operators $\sigma$ and $\rho$ are 
equal if and only if $\textrm{tr}[\sigma W (q,p)]= \textrm{tr}[\rho W(q,p)]$ 
for all $(q,p)\in\R^2$ and the diagonality is equivalent to the condition 
$$
e^{i\phi N}\sigma e^{-i\phi N} =\sigma
$$
for all $\phi\in[0,2\pi)$. Furthermore, since
$$
e^{-i\phi N} W(q,p) e^{i\phi N} =W(q\cos\phi + p\sin \phi, -q\sin\phi +p\cos\phi)
$$
it follows that a state $\sigma $ is diagonal if and only if the mapping
$$
(q,p)\mapsto \textrm{tr} [\sigma W(q,p)]
$$
is invariant with respect to rotations. According to \cite[Prop. 3.4]{Werner} we have 
$$
\textrm{tr}[\mu_{\epsilon_{13},\epsilon_{24}}  *\sigma W(q,p) ] 
= \hat{\mu}_{\epsilon_{13},\epsilon_{24}}(p,-q) \textrm{tr}[\sigma W(q,p)]
$$ 
where 
\begin{align*}
\hat{\mu}_{\epsilon_{13},\epsilon_{24}} (p,-q) =&\int e^{i(px-qy)} 
\, d\mu_{\epsilon_{13},\epsilon_{24}} (x,y)\\
=& \exp\left\{-\frac{1-\epsilon_{24}}{2\epsilon_{24}} q^2 - 
\frac{1-\epsilon_{13}}{2\epsilon_{13}} p^2 \right\}
\end{align*}
is nonzero everywhere. If either of these functions is rotation invariant, 
their product is invariant if and only if the other function is also 
invariant. This proves the Proposition.
\end{proof}
Note that neither of the conditions in Proposition \ref{diagonal} is 
necessary for $\mu_{\epsilon_{13},\epsilon_{24}} *\sigma$ to be diagonal. 
Indeed, consider a state $\sigma =\mu_{\epsilon_{24},\epsilon_{13}}*
\sigma_{\textrm{diag}}$ where $\sigma_{\textrm{diag}}$ is an arbitrary 
diagonal state. For $\epsilon_{13}\neq \epsilon_{24}$ this state is not 
diagonal. On the other hand, since the measure $\mu_{\epsilon_{13}, \epsilon_{24}} 
*\mu_{\epsilon_{24},\epsilon_{13}}$ has the density 
\begin{align*}
(q,p)\mapsto &\: \frac{1}{2\pi} 
\frac{\epsilon_{13}\epsilon_{24}}{\epsilon_{13} -2\epsilon_{13}\epsilon_{24}
+\epsilon_{24}} 
\notag \\ &
\exp\left\{-\frac{1}{2}\frac{\epsilon_{13}\epsilon_{24}}{\epsilon_{13} 
-2\epsilon_{13}\epsilon_{24}+\epsilon_{24}} (q^2 +p^2)\right\}
\end{align*}
it follows from Proposition \ref{diagonal} and the associativity of 
convolutions \cite[Prop. 3.2]{Werner} that 
$$
\mu_{\epsilon_{13},\epsilon_{24}} *(\mu_{\epsilon_{24},\epsilon_{13}} 
*\sigma_{\textrm{diag}} ) =(\mu_{\epsilon_{13},\epsilon_{24}} *
\mu_{\epsilon_{24},\epsilon_{13}} )*\sigma_{\textrm{diag}}
$$ 
is diagonal.
\par
We close this section with a conceptual remark. Since the observable measured 
with this setup is the covariant phase space observable $\G^{\mu_{\epsilon_{13},\epsilon_{24}}*\sigma}$
it is a slight misuse of terminology to call this a {\em direct} measurement of the angle margin 
$\mathsf{\mathsf{\Phi}}^{\mu_{\epsilon_{13},\epsilon_{24}}*\sigma}$. However, the brief analysis below 
shows that this scheme can be used to directly measure $\mathsf{\mathsf{\Phi}}^{\mu_{\epsilon_{13},\epsilon_{24}}*\sigma}$.
Consider for convenience the case of ideal detectors. 
For a local oscillator with a finite intensity $\vert z\vert$  this scheme defines 
an observable $\mathsf{E}^\sigma_z :\h B(\R^2)\rightarrow \lh$. It was shown in \cite{Kiukas} that, 
with the choice $\textrm{arg}(z)=0$, 
$$
\lim_{\vert z\vert \rightarrow \infty} \mathsf{E}^\sigma_z =\G^\sigma
$$
weakly in the sense of probabilities (see \cite{Kiukas} for details). Now  $\mathsf{E}^\sigma_z$ 
is a discrete observable and the measurement outcomes consist of pairs $(q,p)\in\R^2$. 
Let $f:\R^2\setminus \{(0,0)\} \rightarrow [0,2\pi)$ be the pointer function which assigns to each pair the corresponding 
argument, that is, $f(q,p) =\alpha_{qp}$ defined by 
\begin{equation*}
 \cos \alpha_{qp} =\frac{q}{\sqrt{q^2+p^2}},\qquad \sin\alpha_{qp} =\frac{p}{\sqrt{q^2+p^2}}
\end{equation*}
and denote $\mathsf{E}^{f,\sigma}_z :\h B([0,2\pi))\rightarrow \lh$, 
\begin{equation*}
 \mathsf{E}^{f,\sigma}_z (X) =\mathsf{E}^{\sigma}_z \big(f^{-1}(X)\cup \{(0,0)\}\big)
\end{equation*}
Then it can be shown that 
\begin{equation*}
 \lim_{\vert z\vert\rightarrow \infty} \mathsf{E}^{f,\sigma}_z = \mathsf{\mathsf{\Phi}}^\sigma
\end{equation*}
weakly in the sense of probabilities and the same argumentation holds in the case of 
inefficient detectors. In this sense, by {\em choosing} to record only the values $\alpha_{qp}$ 
we see that eight-port homodyne detection in the high-amplitude limit can be used as a direct measurement of  
$\mathsf{\mathsf{\Phi}}^{\mu_{\epsilon_{13},\epsilon_{24}}*\sigma}$.
\section{Balancing efficiencies by an additional beam splitter}\label{beamsplitter}
Suppose that the state of the parameter field is diagonal, for instance,
a vacuum state. In order to obtain a CPO, we need 
to have $\epsilon_{13}=\epsilon_{24}$. As illustrated in Fig. \ref{efficiency},
a given value of $\epsilon_{ij}$ can be obtained with infinitely many
different values of $\epsilon_i$ and $\epsilon_j$. It follows that there
is a great deal of freedom in choosing the detectors in order to obtain
the equality $\epsilon_{13}=\epsilon_{24}$. This degree of freedom may
be exploited to modify the measurement setup in order to
compensate any difference in the overall efficiencies. Indeed, suppose
that the efficiencies $\epsilon_j$ are fixed and, for instance,
$\epsilon_{24}<\epsilon_{13}$. This means that the homodyne detector
consisting of detectors $D_1$ and $D_3$ is more efficient than the other
one. 
\begin{figure}[h!]
\includegraphics[width=0.95\columnwidth]{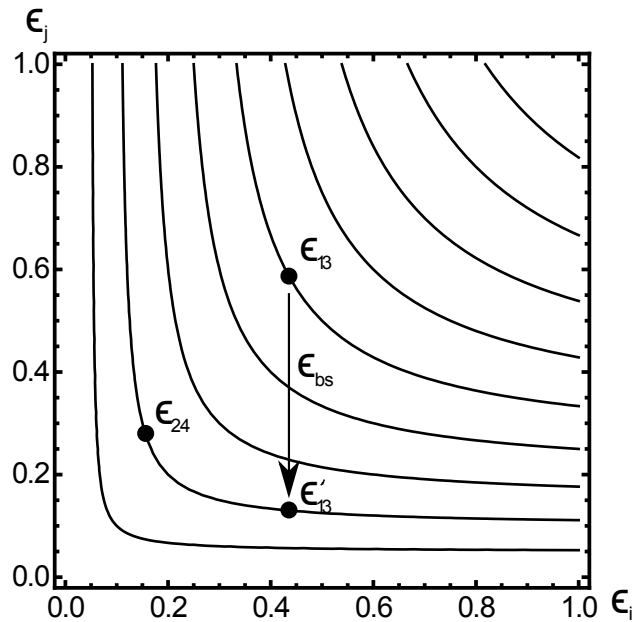}
\caption{Contourlines for the overall efficiency $\epsilon_{ij}$ as 
a function of $\epsilon_i$ and $\epsilon_j$. The addition of a beam splitter 
with transparency $\epsilon_{\textrm{bs}}$ can be used to balance the 
setup and and obtain $\epsilon'_{13}=\epsilon_{24}$.}\label{efficiency}
\end{figure}
\par
Since a photodetector with efficiency $\epsilon$ is equivalent to having  
a fictitious beam splitter with transparency $\epsilon$ in front of an 
ideal detector (see, e.g. \cite{Leonhardt,Busch}), one can artificially decrease the 
efficiency of, say, detector $D_3$ by placing an additional beam splitter 
with transparency $\epsilon_{\textrm{bs}}$ in front of the detector. The 
resulting effective efficiency of $D_3$ is then $\epsilon_3'
=\epsilon_{\textrm{bs}}\epsilon_3$ and the new overall efficiency is 
$$
\epsilon'_{13}= \frac{2\epsilon_{\textrm{bs}}\epsilon_1
\epsilon_3}{\epsilon_1 +\epsilon_{\textrm{bs}}\epsilon_3}
$$
Hence, with the appropriate choice
$$
\epsilon_{\textrm{bs}} =\frac{\epsilon_1 \epsilon_{24}}{2
\epsilon_1\epsilon_3 -\epsilon_3\epsilon_{24}}
$$
we may balance the setup and obtain $\epsilon'_{13}=\epsilon_{24}$.
This is illustrated in Fig. \ref{efficiency}. In other words, we
achieve a CPO at the price of artificially
decreasing the largest efficiency to the value of the smallest one.
For the remainder of this section we denote $\epsilon =\epsilon_{13}'
=\epsilon_{24}$ and use the notation $\mu_\epsilon =\mu_{\epsilon, 
\epsilon}$.
\par
It is interesting to note that by balancing the efficiencies of the
homodyne detectors we have a situation where both the actually measured
observable and the one corresponding to ideal detectors are phase
observables. Therefore it is natural to study the connection between
them. Since the measured phase-space observable is a smearing of the
ideal one, one might expect that this property is inherited into the
angle margins, namely, that there exists a probability measure $\nu:\h B
([0,2\pi))\rightarrow [0,1]$ such that $\mathsf{\Phi}^{\mu_\epsilon
*\sigma} =\nu *\mathsf{\Phi}^\sigma$.  However, this is not the case. 
\begin{proposition}
The measured observable $\mathsf{\Phi}^{\mu_\epsilon *\sigma}$ is never a smearing of
$\mathsf{\Phi}^\sigma$. 
\end{proposition}
\begin{proof}
Assume that $\mathsf{\Phi}^{\mu_\epsilon *\sigma} =\nu *\mathsf{\Phi}^\sigma$
for some probability measure $\nu$. Let $(c_{mn}^{\mu_\epsilon *\sigma})$ and 
$(c_{mn}^\sigma )$ denote the phase matrices of $\mathsf{\Phi}^{\mu_\epsilon *\sigma}$ 
and $\mathsf{\Phi}^\sigma$, respectively. It is easily verified using Eq. \eqref{phase} 
that the matrix elements satisfy the relation
\begin{equation}\label{matrixelements}
c_{m,m+k}^{\mu_\epsilon *\sigma} =\hat{\nu}(k) c_{m,m+k}^\sigma
\end{equation}
where $\hat{\nu}(k)= \int_0^{2\pi} e^{ik\alpha}\, d\nu(\alpha)$. It was shown in \cite{Lahti3} that 
$\lim_{m\rightarrow \infty} c_{m,m+k}^{\vert n\rangle} =1$ for all $k\in\N$, where 
$(c_{m,m+k}^{\vert n\rangle})$ is the phase matrix related to the observable $\mathsf{\Phi}^{\vert n\rangle}$. 
Now both $\sigma$ and $ \mu_\epsilon *\sigma$ are mixtures of number states and
the convex structure is inherited into the corresponding observables, and thus into 
the phase matrices. Therefore, we have
\begin{equation*}
 \lim_{m\rightarrow \infty} c_{m,m+k}^{\mu_\epsilon *\sigma} =1=\lim_{m\rightarrow \infty}c_{m,m+k}^\sigma
\end{equation*}
for all $k\in \N$. This, together with Eq. \eqref{matrixelements} shows that $\hat{\nu} (k)=1$ for all $k\in\N$. 
It follows that $\mathsf{\Phi}^{\mu_\epsilon *\sigma} =\mathsf{\Phi}^\sigma$ which is possible if and only if 
$ \mu_\epsilon *\sigma =\sigma$. This is satisfied if and only if $\epsilon=1$, that is, the detectors are ideal.
\end{proof}
In the simplest case of the vacuum parameter field and balanced 
efficiencies the convoluted state can easily be calculated. First notice 
that  the necessary matrix elements of the Weyl operators are 
$$
\langle n\vert W(q,p) \vert 0\rangle = \frac{1}{\sqrt{n!}} 
\left( \frac{1}{\sqrt{2}} (q+ip)\right)^n e^{-\frac{1}{4}(q^2+p^2)}
$$
so that with the polar coordinates $re^{i\alpha} =\frac{1}{\sqrt{2}} (q+ip)$, 
one can calculate
\begin{align*}
\langle n\vert \mu_{\epsilon} * \vert 0\rangle\langle 0\vert \vert n\rangle 
&= \int \big\vert\langle n\vert W(q,p) \vert 0\rangle \big\vert^2 \, 
d\mu_{\epsilon} (q,p)\\
&= \frac{1}{n!} \frac{\epsilon}{1-\epsilon} \int r^{2n} 
\exp\left\{-\frac{r^2}{1-\epsilon}\right\} 
\frac{dr^2 d\alpha}{2\pi} \\
&= \epsilon (1-\epsilon)^n
\end{align*}
The convoluted state is thus
\begin{equation}\label{state1}
 \mu_{\epsilon} *\vert 0\rangle\langle 0\vert  
 =\epsilon \sum_{n=0}^\infty (1-\epsilon)^n 
 \vert n\rangle\langle n\vert
\end{equation}
\section{Balancing efficiencies by squeezing the parameter field}\label{squeezing}
There is an interesting alternative to the method of balancing
efficiencies considered above. As mentioned before, the requirement of
equal efficiencies is necessary only in the case that the parameter
field is in a diagonal state. Therefore it is possible that for fixed
efficiencies a suitably chosen non-diagonal state can be used to
compensate for the difference in the efficiencies so that the convoluted
state is diagonal. Here we show that this can always be done by
employing a suitable squeezed vacuum state as a parameter field.
\par
Let us assume that we are able to prepare the parameter field into a 
squeezed vacuum state $\vert \psi_a \rangle\langle \psi_a \vert $, 
where $a>0$ is the squeezing parameter and $\psi_a (x) 
=\left(a/\pi\right)^{1/4} e^{-\frac12 ax^2}$. As in the 
proof of Prop. \ref{diagonal}, we need to study the rotation 
invariance of the function
\begin{align}
(q,p)&\mapsto \textrm{tr} \left[\mu_{\epsilon_{13},\epsilon_{24} }* 
\vert \psi_a \rangle\langle \psi_a \vert W(q,p)\right] \nonumber\\
=\:&\hat{\mu}_{\epsilon_{13},\epsilon_{24}} (p,-q) 
\label{eq}
\langle \psi_a \vert W(q,p) \vert \psi_a\rangle \\
=\:& e^{-\left(\frac{1-\epsilon_{24}}{2\epsilon_{24}}+\frac{a}4\right) q^2-
\left(\frac{1-\epsilon_{13}}{2\epsilon_{13}}+\frac1{4a}\right)
p^2}
\notag
\end{align}
It is clear that this is invariant with respect to rotations if we can 
choose the squeezing parameter in such a way that the equality
\begin{equation}\label{eq3}
\frac{1-\epsilon_{24}}{2\epsilon_{24}} +\frac{a}{4} =
\frac{1-\epsilon_{13}}{2\epsilon_{13}} +\frac{1}{4a}
\end{equation}
holds. Solving Eq. \eqref{eq3} for $a$ we have
\begin{equation}\label{a}
a =\frac{\epsilon_{24}-\epsilon_{13}}{\epsilon_{13}\epsilon_{24}} 
\pm \sqrt{1+\left(\frac{\epsilon_{24}-\epsilon_{13}}{\epsilon_{13}\epsilon_{24}}\right)^2}
\end{equation}
where the solution with the plus sign is always positive. Hence, we
can compensate the difference in the efficiencies by using a suitably
squeezed vacuum as the parameter field. In order to compare this with
the method of balancing efficiencies we need to solve the spectral
decomposition of the convoluted state
$$
\mu_{\epsilon_{13},\epsilon_{24} }* 
\vert \psi_a \rangle\langle \psi_a \vert
$$ 
\par
First, define a parameter
$$
\eta =\frac{\epsilon_{13}-2\epsilon_{13}\epsilon_{24}+
\epsilon_{24}}{\epsilon_{13}\epsilon_{24}} +  
\sqrt{1+\left(\frac{\epsilon_{24}-\epsilon_{13}}
{\epsilon_{13}\epsilon_{24}}\right)^2}
$$
so that by inserting the value \eqref{a} of the squeezing parameter 
into Eq. (\ref{eq}) we obtain
\begin{equation}
\textrm{tr} [\mu_{\epsilon_{13},\epsilon_{24} }* 
\vert \psi_a \rangle\langle \psi_a \vert W(q,p)] =
e^{-\frac{\eta}{4} (q^2 +p^2)} \label{eq1}
\end{equation}
On the other hand we know that 
$$
\mu_{\epsilon_{13},\epsilon_{24}} * \vert \psi_a\rangle\langle 
\psi_a\vert =\sum_{n=0}^\infty \lambda_n \vert n\rangle\langle 
n\vert 
$$
so that 
\begin{align}
\textrm{tr}&\left[ \mu_{\epsilon_{13},\epsilon_{24}} * 
\vert \psi_a\rangle\langle \psi_a\vert W(q,p)\right]
= \sum_{n=0}^\infty \lambda_n \langle n\vert W(q,p) 
\vert n\rangle \nonumber \\ 
&= \sum_{n=0}^\infty \lambda_n e^{-\frac{1}{4} (q^2 +p^2)} 
L_n \left(\tfrac{1}{2}(q^2 +p^2)\right) \label{lambda1}
\end{align}
where
$L_n (x)$ denotes the $n$-th Laguerre polynomial. 
Upon rewriting the  exponential function 
in Eq. (\ref{eq1}) using the series representation 
\cite[8.975(1)]{Gradshteyn}
$$
 e^{\frac{z}{z-1}x} =(1-z)\sum_{n=0}^\infty L_n (x) z^n,\qquad \vert z\vert <1,
$$
one has
\begin{align}
e^{-\frac{\eta}{4}(q^2 +p^2) } &= 
e^{-\frac{1}{4}(q^2 +p^2) } e^{\frac{1}{4}(1-\eta)(q^2 +p^2)} \notag
\\ &= \frac{2e^{-\frac{1}{4}(q^2 +p^2) }}{\eta +1} 
\sum_{n=0}^\infty \left(\tfrac{\eta-1}{\eta +1} \right)^n 
L_n \left(\frac{q^2 +p^2}2\right)\label{lambda2}\,,
\end{align}
where $0<\frac{\eta-1}{\eta +1}<1$. Comparing 
Eqs. \eqref{lambda1} and \eqref{lambda2} we find that 
the eigenvalues $\lambda_n$ are
$$
\lambda_n =\frac{2}{\eta +1} \left(\frac{\eta-1}{\eta +1} \right)^n
$$
and the state is 
\begin{equation}\label{state2}
\mu_{\epsilon_{13},\epsilon_{24}} *\vert \psi_a
\rangle\langle \psi_a \vert =\epsilon_{\textrm{eff}} 
\sum_{n=0}^\infty (1-\epsilon_{\textrm{eff}})^n \vert n\rangle\langle n\vert
\end{equation}
where we have defined $$\epsilon_{\textrm{eff}} =\frac{2}{\eta +1}\,,$$ 
which may be viewed as the overall effective efficiency of this measurement scheme. 
\par
The remarkable feature of this method is the inequality 
\begin{equation}\label{comparison}
\epsilon_{\textrm{eff}} \geq \epsilon_{\textrm m} 
\equiv \min \{ \epsilon_{13},\epsilon_{24} \}
\end{equation}
which holds for any value of the quantum efficiencies. Furthermore, the equality 
holds if and only if $\epsilon_{13} =\epsilon_{24}$ and in this case no squeezing 
is needed. This means that for $\epsilon_{13} \neq \epsilon_{24}$ the overall 
efficiency of this method is always greater than the one obtained by balancing 
the efficiencies by the insertion of an additional beam splitter. Indeed, by 
multiplying both sides of \eqref{comparison} by $ (\eta +1)\max\{ \epsilon_{13},
\epsilon_{24}\}$ and after some algebra we see that \eqref{comparison} is equivalent to  
\begin{equation*}
\sqrt{\epsilon_{13}^2\epsilon_{24}^2 +(\epsilon_{24}-\epsilon_{13})^2} 
\leq  \vert \epsilon_{24}- \epsilon_{13} \vert +\epsilon_{13}\epsilon_{24}
\end{equation*}
which holds for all $\epsilon_{13}$ and  $ \epsilon_{24}$.
\par
In order to make our analysis more quantitative let us introduce the 
quantity
\begin{align}
\gamma = 
\frac{\epsilon_{\textrm{eff}}}{\epsilon_{\textrm m}} = 
\frac{2}{(1+\eta)\epsilon_{\textrm m}}\,, 
\end{align}
which represents the ratio between the effective efficiency achievable
by squeezing the parameter field at fixed value of the four efficiencies 
$\epsilon_j$, $j=1,..,4$ and the corresponding quantity obtained by 
the insertion of a beam splitter. From Eq. (\ref{comparison}) we know already 
that $\gamma\geq1$, whereas in Fig. \ref{f:gm} we report its behaviour as a
function of $\epsilon_{13}$ and $\epsilon_{24}$. 
\begin{figure}[h]
\includegraphics[width=0.95\columnwidth]{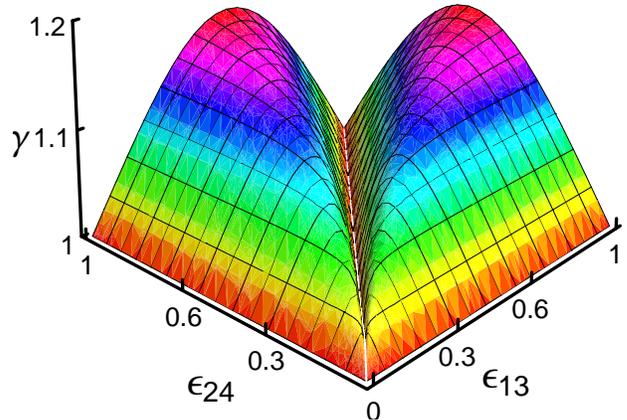}
\caption{\label{f:gm} (Color Online) The ratio $\gamma$ between 
the effective efficiency achievable
by squeezing the parameter field 
and the corresponding quantity obtained by 
the insertion of a beam splitter as a
function of $\epsilon_{13}$ and $\epsilon_{24}$.}
\end{figure}
\par
As it is apparent from the plot $\gamma$ is symmetric under the exchange
of $\epsilon_{13}$ and $\epsilon_{24}$ and achieves its maximum
$\gamma\simeq 1.17$ for $\epsilon_{13}=0.5$ and $\epsilon_{24}=1$ 
or viceversa. The function is not particularly peaked around its maximum
and this means that there is a wide range of values for 
$\epsilon_{13}$ and $\epsilon_{24}$ for which we have a significant gain 
in squeezing the parameter field in comparison to the insertion of a
beam splitter. On the other hand, when one of the two efficiencies is
very small then the two methods are equally ineffective.
The amount of squeezing needed to achieve CPO strongly depends on the 
values of the efficiencies. The region of maximum improvement
corresponds to a moderate squeezing, i.e. $a$ not too far from one.
In Fig. \ref{f:pargm} we report the parametric plot of $\gamma$ as a
function of the corresponding squeezing: this is a multivalued plot since
there are many pairs $(\epsilon_{13},\epsilon_{24})$ for which the same
$\gamma$ is achievable, though employing different amounts of squeezing.
\begin{figure}[h]
\includegraphics[width=0.95\columnwidth]{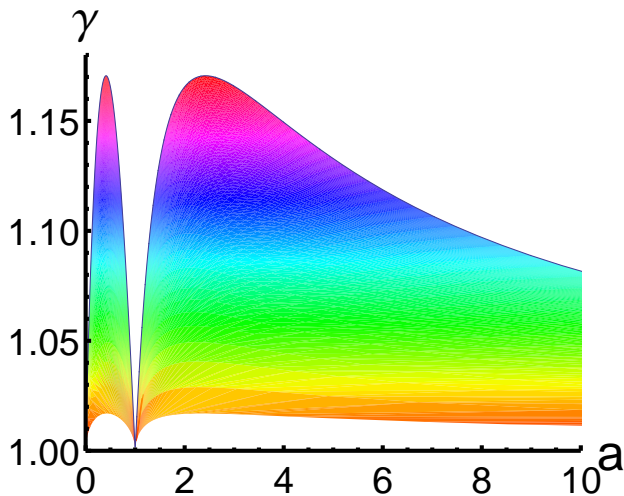}
\caption{\label{f:pargm} (Color Online) Parametric plot of 
the ratio $\gamma$ between the effective efficiency achievable
by squeezing the parameter field and the corresponding quantity 
obtained by the insertion of a beam splitter as a
function of the squeezing $a$ needed to achieve the compensation.
}\end{figure}
\par
The two symmetric maxima of Fig. \ref{f:gm} correspond to 
squeezing parameters which are inverses of each other, ($a\simeq 2.414$
and $a\simeq 0.414$) i.e. they correspond to the same amount of squeezing, 
but in orthogonal directions. In turn, when the values of the two efficiencies 
$\epsilon_{13}$ and $\epsilon_{24}$ are close to each other we have
\begin{align*}
a & \simeq 1 + \frac{\epsilon_{13}-\epsilon_{24}}{\epsilon_{\textrm
m}^2}   
+ \h O (\epsilon_{13}-\epsilon_{24})^2\,,
\\ 
\epsilon_{\textrm{eff}} & \simeq \epsilon_{\textrm m}+ \frac12 |\epsilon_{13}-\epsilon_{24}|
+ \h O (\epsilon_{13}-\epsilon_{24})^2\,,
\\ 
\gamma & \simeq 1 + \frac{|\epsilon_{13}-\epsilon_{24}|}{2 \epsilon_{\textrm m}} 
+ \h O (\epsilon_{13}-\epsilon_{24})^2\,.
\end{align*}
Overall, we conclude that squeezing the parameter field is always convenient,
and may lead to a considerable gain in the effective efficiency in comparison  
to the insertion of a beam splitter. Since the maximum gain corresponds to the
use of a moderate amount of squeezing we foresee possible experimental implementations
with current technology.
\par
We close this section by comparing the phase distributions obtained by
using these two methods of balancing the efficiencies. Suppose, for
simplicity, that the signal field is in a coherent state $\vert
z\rangle$ with $z =1$ and the overall efficiencies of the homodyne
detectors are $\epsilon_{13}=0.5$ and $\epsilon_{24}=1$. The effective
efficiency obtained by using the squeezing method is then
$\epsilon_{\textrm{eff}} \simeq 0.828$. We show the phase
distributions in Fig. \ref{distributions}  where we have also added the
ideal case for comparison. It is clear that the squeezing method
provides a distribution which is more peaked around its maximum. To make
this more precise, consider the {\em minimum variance} of the
distribution defined as \cite{Lahti3} 
\begin{equation*}
\textrm{Var}_\textrm{min}(p) =\inf_{\phi,\varphi\in\R} \, \frac{1}{2\pi} \int_{\varphi-\pi}^{\varphi+\pi} (\phi-\varphi)^2p(\phi) \, d\phi
\end{equation*}
Let $p_{\textrm{id}}$, $p_{\textrm{sq}}$ and $p_{\textrm{bs}}$ be the
phase distributions obtained with ideal detectors, squeezing, and by
using an additional beam splitter. Then the minimum variances are given
by \begin{align*}
&\textrm{Var}_\textrm{min}(p_\textrm{id})\simeq 0.76,\\ 
&\textrm{Var}_\textrm{min}(p_\textrm{sq}) \simeq 0.89,\\
&\textrm{Var}_\textrm{min}(p_\textrm{bs}) \simeq 1.24, 
\end{align*}
which clearly shows the advantage of the squeezing method.
\begin{figure}[h!]
\includegraphics[width=0.95\columnwidth]{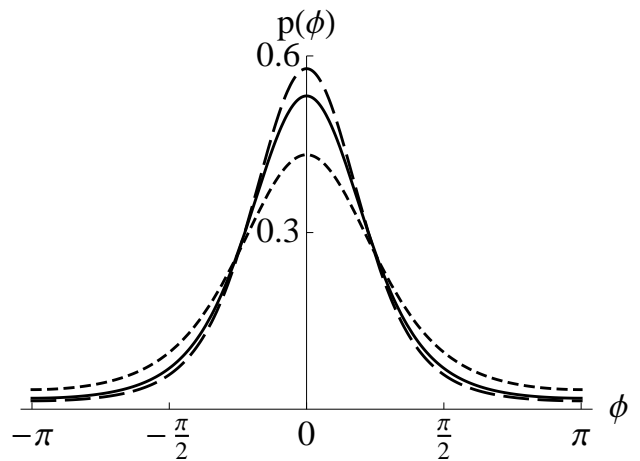}
\caption{The phase distributions of a coherent state $\vert z\rangle$ with $z=1$ in the case of ideal detectors (higher dashed line),  balancing by squeezing (solid line), and balancing by an additional beam splitter (lower dashed line).}\label{distributions}
\end{figure}
\section{Conclusions and outlooks}\label{conclusions}
In this paper we have analyzed in detail the performance of the 
eight-port homodyne detector as a suitable device to measure 
covariant phase observables. We have abandoned the traditional 
assumption of equal quantum efficiencies for the four photodetectors 
involved in the detection scheme and have investigated in detail 
the conditions under which the measurement of a CPO
may be achieved. We have found that balancing the efficiencies using an
additional beam splitter allows to achieve CPO at the price of reducing
the overall effective efficiency and we have proved that this CPO is never 
a smearing of the ideal CPO achievable with unit quantum efficiency.
We have also suggested an alternative compensation strategy, where a 
squeezed vacuum is used as a parameter field, which allows 
one to increase the overall efficiency in comparison to the passive case 
using only a moderate amount of squeezing. 
\par
In ideal conditions, i.e. for photodetectors with unit quantum efficiencies, 
the phase-space observables achievable by eight-port homodyning are
equivalent to those achievable by six-port homodyning \cite{Par97}
or heterodyning \cite{Yue78,Sha84,Sha85}. Equivalence also holds in noisy conditions 
if all the involved photodetectors are assumed to have the same quantum 
efficiency. In this context a question arises on whether the effects of
different quantum efficiencies may result in different phase-space
observables or in inequivalent compensation schemes.
Work along these lines is in progress and results will be reported
elsewhere.
\par
Our results provide a more realistic characterization of phase-space 
measurements of the optical phase by eight-port homodyning, and are suitable 
for experimental verification. As a matter of fact, both compensation 
schemes suggested in this paper may be implemented with current quantum optical 
technology. 
\section*{Acknowledgment} 
JS and JPP thank Pekka Lahti for useful discussions. JS is grateful to
 the Finnish Cultural Foundation for financial support.
MGAP thanks Sabrina Maniscalco for several discussions and 
the Finnish Cultural Foundation (Science Workshop on Entanglement) 
for financial support.

\end{document}